\documentclass[conference]{IEEEtran}

\usepackage[font=small,skip=2pt]{caption}

\usepackage[bookmarks=false]{hyperref} 

\usepackage{url}
\usepackage{color}
\usepackage{algorithm}
\usepackage{algpseudocode}
\usepackage{textcomp}
\usepackage{cite}
\usepackage{graphicx}
\usepackage{amssymb,hhline,enumerate,dsfont}
\usepackage{amsmath}
\usepackage{array}
\usepackage{psfrag}
\usepackage{epsfig}
\usepackage{amsthm}
\usepackage{amsmath,amssymb,times}
\usepackage{graphicx,multirow,epsfig,cite}

\def\block(#1,#2)#3{\multicolumn{#2}{c}{\multirow{#1}{*}{$ #3 $}}}

\newtheorem{thm}{Theorem}
\newtheorem{rmrk}{Remark}
\newtheorem{lem}{Lemma}

\newtheorem{defn}{Definition}

\IEEEoverridecommandlockouts

\begin{document}
\title{Product Matrix Minimum Storage Regenerating Codes with Flexible Number of Helpers}

\author{\IEEEauthorblockN{Kaveh Mahdaviani}
\IEEEauthorblockA{ECE Dept., University of Toronto\\
Toronto, ON M5S3G4, Canada\\
Email: kaveh@comm.utoronto.ca}
\and
\IEEEauthorblockN{Soheil Mohajer}
\IEEEauthorblockA{ECE Dept., University of Minnesota\\
Minneapolis, MN 55404, USA\\
Email: soheil@umn.edu}\thanks{The work of S. Mohajer is supported by the National Science Foundation under Grant CCF-1617884.}
\and
\IEEEauthorblockN{Ashish Khisti}
\IEEEauthorblockA{ECE Dept., University of Toronto\\
Toronto, ON M5S3G4, Canada\\
Email: akhisti@comm.utoronto.ca}}

\maketitle
\vspace{-3.5cm}

\begin{abstract}

In coding for distributed storage systems, efficient data reconstruction and repair through accessing a predefined number of arbitrarily chosen storage nodes is guaranteed by regenerating codes. Traditionally, code parameters, specially the number of helper nodes participating in a repair process, are predetermined. However, depending on the state of the system and network traffic, it is desirable to adapt such parameters accordingly in order to minimize the cost of repair. In this work a class of regenerating codes with minimum storage is introduced that can simultaneously operate at the optimal repair bandwidth, for a wide range of exact repair mechanisms, based on different number of helper nodes.

\end{abstract}

\IEEEpeerreviewmaketitle

\section{Introduction}

For a distributed storage system (DSS) reliability and accessibility are the most important features. With large scale DSSs nowadays it is common to lose access to a storage node or part of its content. Hence, both reliability and accessibility depend on system's capability to replace a failed node by a new one, and recover its content. This procedure is referred to as \emph{repair}. For the DSS to be capable of repair, it is necessary to store redundancy. It is shown that there exits an information theoretic tradeoff between the amount of redundancy (\emph{i.e.} storage overhead), and the amount of data transmission required for a repair, referred to as \emph{repair bandwidth} \cite{Regenerating}. 

Among various methods of storing redundancy and performing repair, a specific class of erasure codes, named regenerating codes, offers the efficient performance \cite{Regenerating}. More precisely, a regenerating code on a Galois field $\mathbb{F}_{q}$ for a DSS with $n$ storage nodes, maps the source data of size $F$ symbols into $n$ pieces of size $\alpha$ symbols each, and stores them in $n$ separate nodes, such that any $k$ out of $n$ nodes suffice to recover the data. Such system is capable to tolerate up to $(n-k)$ node failures. Moreover, upon failure of one node, it can be replaced by a new node whose content is determined by connecting to an arbitrary set of $d$ (where $d\geq k$) helper nodes, and downloading $\beta$ symbols form each (where $\beta \leq \alpha$). Ideally, one would like to minimize the storage overhead, and repair bandwidth simultaneously. It turns out that for a given file size $F$, there is a tradeoff between the per-node storage capacity $\alpha$ and the repair bandwidth $\gamma=d\beta$, and one can be minimized only at the cost of a penalty for the other \cite{Regenerating}. In  particular, at one extreme point of this tradeoff, one could first minimize the \emph{per-node storage}, $\alpha$, and then minimize the \emph{per-node repair bandwidth}, $\beta$, to obtain a \emph{minimum storage regenerating} (MSR) code. As a result, MSR codes have the \emph{maximum distance separable} (MDS) property, and also minimize the repair bandwidth for the given $\alpha$ \cite{Regenerating}, which means for an MSR code we have $F=k\alpha$, and 
\begin{align}\label{eq_betaMSR}
\beta = \frac{F}{k(d-k+1)}.
\end{align}
In other words, in MSR codes the repair bandwidth is decreasing super-linearly as $d$ grows. 

Reversing the order of minimization between $\alpha$, and $\beta$ results in another extreme point of the tradeoff, which provides the minimum repair bandwidth (MBR) regenerating codes. Our focus in this work is on MSR codes as they minimize the storage cost. Moreover, we only consider the repair mechanisms in which the replacement node contains exactly the same content as stored in the failed node. Such repair mechanisms are referred to as \emph{exact repair}, and enable the code to be systematic, which is a significant advantage in practice.




The common adopted model in regenerating codes considers a predetermined number $d$ (where $k\leq d \leq n-1$) of helpers required for any repair procedure. Each of these helpers is also assumed to provide $\beta = \gamma/d$ repair bandwidth. This sets a threshold for the system's capability to perform repair. On the other hand, in practice the state of system dynamically changes as a function of various factors including traffic load, available bandwidth, \emph{etc.} Therefore, runtime adaptation would be of great value towards optimizing the performance. For instance, when the system is heavily loaded by many read requests, there might be only very few nodes available to serve as helpers. In this situation we are interested in optimal repair based on the available helpers. Likewise, when there are many helpers available it is beneficial to use a large number of helpers as increasing $d$ reduces both $\gamma$ and $\beta$ in optimal repair mechanism characterized by equation (\ref{eq_betaMSR}). This could then reduce both the total network traffic as well as the transmission delay. We refer to such property as \emph{bandwidth adaptive}. 

The design of such codes has been of interest and the significance of bandwidth adaptivity in the performance of the system has been emphasised in \cite{BWAdapt_Kermarrec ,Exact_AI_Asymptotic, BWAdapt_Opportunistic, ProgressiveEngagement}. However, it is a challenging problem to design such coding scheme with a large flexibility degree since it needs to satisfy many optimality conditions simultaneously. As a result, this problem has only been considered for the MSR \cite{Exact_AI_Asymptotic, ExplicitMSR}, and MBR \cite{BAER_ISIT16, BAER} extreme points of the tradeoff. For the MBR case, \cite{BAER_ISIT16, BAER} provided a solution for a wide range of practical parameters based on the \emph{Product Matrix} framework introduced in \cite{PM_Codes}. In \cite{Exact_AI_Asymptotic} a solution is provided based on interference alignment, which only achieves the MSR characteristics when both $\alpha$ and $\beta$ tend to infinity. The first explicit exact repair MSR code constructions which satisfy the bandwidth adaptivity are introduced in \cite{ExplicitMSR}. These constructions work for any parameters $k$, $n$, and all values of $d$ such that $k<d<n$. Although these constructions can achieve optimality for finite values of $\alpha$ and $\beta$, but the required value for these parameters are still very huge (\emph{i.e.} exponentially large in $n$), and hence they only achieve optimality for extremely large contents. Recently, \cite{ExplicitMSR_nearOptimal} introduced a modified version of the codes in \cite{ExplicitMSR} which achieves MSR optimality for much lower values of $\alpha$, at the cost of loosing bandwidth adaptivity. Indeed the MSR code in \cite{ExplicitMSR_nearOptimal} works only for $d=n-1$. In \cite{Coupled_Layer}, $d<n-1$, and practical $\alpha$ is achieved for MSR codes without bandwidth adaptivity.

In this work we address the design of MSR codes with bandwidth adaptive exact repair for small $\alpha$, and $\beta$, following the Product Matrix framework \cite{PM_Codes}. The code allows us to choose the number of helper nodes for each repair scenario independently, and it is capable to adjust the per-node repair bandwidth to its optimum value based on the number of selected helpers as in (\ref{eq_betaMSR}). Compared to the constructions proposed in \cite{ExplicitMSR} for a DSS with $n$ storage nodes the required values for  $\alpha$ and $\beta$ in the presented code is reduced to the $n^{\text{th}}$ root for the same set of other parameters. The main contributions of this work are explained in the next section, after formally defining the problem setup.

\section{Model and Main Results}

\subsection{Model}

The first element we consider for the model of our bandwidth adaptive distributed storage system is a predefined Galois field alphabet, $\mathbb{F}_{q}$ of size $q$. Hereafter we assume all the symbols stored or transmitted through the network are elements of $\mathbb{F}_q$. Besides, we will consider a homogeneous group of $n$ storage nodes, each capable of storing $\alpha$ symbols.

\begin{defn}[Bandwidth Adaptive Regenerating Code]
Consider the set of parameters $\alpha$, $n$, $k$, $\delta$, a set $D=\{d_{1}, \cdots, d_{\delta}\}$, with $d_{1}<\cdots<d_{\delta}$, and a \emph{total repair bandwidth function} $\gamma:D\rightarrow [\alpha,\infty)$. A bandwidth adaptive regenerating code $\mathcal{C}(n$, $k$, $D$, $\alpha$, $\gamma)$ is a regenerating code with per-node storage capacity $\alpha$, such that in each repair process the number of helpers, $d$, can be chosen arbitrarily from the set $D$. The choice of helper nodes is also arbitrary, and each of the chosen helpers then provides $\beta(d) = \gamma(d) / d$ repair symbols. Moreover, the data collector recovers the whole source data by accessing any arbitrary set of $k$ nodes.
\end{defn}

Note that the flexibility of the repair procedure depends on the parameter $\delta$, such that for a larger $\delta$, there are more options to select the number of helpers. In general, it is appealing to have small choices such as $d_{1}$, to guarantee the capability of code to perform repair when the number of available helpers is small, and also large choices such as $d_{\delta}$, to provide the capability of reducing the per-node repair bandwidth and hence the transmission delay whenever a larger number of helpers are available. The coding scheme we present in this work allows to design such a range for the elements in $D$.

\begin{defn}[Total Storage Capacity]
For the set of parameters $\alpha$, $n$, $k$, $\delta$, a set $D=\{d_{1}, \cdots, d_{\delta}\}$, and a given function $\gamma:D \rightarrow [\alpha,\infty)$, the \emph{total storage capacity} of a bandwidth adaptive distributed storage system is the maximum size of a file that could be stored in a network of $n$ storage nodes with per-node storage capacity $\alpha$, using a bandwidth adaptive regenerating code $\mathcal{C}(n$, $k$, $D$, $\alpha$, $\gamma)$. We will denote the storage capacity of such a system by $F(n$, $k$, $D$, $\alpha$, $\gamma)$, or simply $F$ when the parameters could be inferred from the context.
\end{defn}

\begin{defn}[Bandwidth Adaptive MSR Codes, and the Flexibility Degree]
For any choice of parameters $\alpha$, $n$, $k$, $\delta$, and set $D=\{d_{1}, \cdots, d_{\delta}\}$, the bandwidth adaptive regenerating codes that realize both the MDS property defined by $F(n$, $k$, $D$, $\alpha$, $\gamma) = k \alpha$, as well as the the MSR characteristic equation simultaneously for all $d\in D$, given as,
\begin{align}\label{eq_BWAMSR}
\alpha = (d-k+1)\beta(d), ~~ \forall{d \in D},
\end{align}
is referred to as \emph{bandwidth adaptive MSR} codes. Moreover, the number of elements in the set $D$ is referred to as \emph{flexibility degree} of the code, and is denoted by $\delta$.
\end{defn}


\subsection{Main Results}

The main contribution of this work is to provide a bandwidth adaptive MSR coding scheme with small per-node storage requirement. This coding scheme also guarantees exact repair for different choices of the number of helpers. This result is formally stated in the form of the following theorem. In this paper $\mathrm{lcm}()$ denotes the least common multiple.

\begin{thm}\label{Thm_main}
For arbitrary positive integers $n$, $k$, and $\delta$, there exists an adaptive bandwidth MSR code, with a finite per-node storage capacity $\alpha$ and total storage capacity $F$, satisfying 
\begin{align}
\alpha = (k-1)\mathrm{lcm}\left(1,2,\cdots,\delta\right), ~~~~ F = k \alpha, \nonumber
\end{align}
which is capable of performing exact repair using any arbitrary $d_{i}$ helpers, for 
\begin{align}
d_{i}=(i+1)(k-1),~i \in \{1, \cdots, \delta\}, \nonumber
\end{align}
and simultaneously satisfies the MSR characteristic equation (\ref{eq_BWAMSR}) for any $d_{i}$. \emph{i.e.,}
\begin{align}
\beta(d_{i}) = \frac{\alpha}{(d_{i}-k+1)},~i \in \{1, \cdots, \delta\}. \nonumber
\end{align}
\end{thm}

Section \ref{Sec_Coding_Scheme} provides a constructive proof for this theorem.

\section{Coding Scheme}\label{Sec_Coding_Scheme}
The coding scheme presented in this work is closely related to the Product Matrix MSR code introduced in \cite{PM_Codes}, and could be considered as an extension of the Product Matrix code, that achieves bandwidth adaptivity. To demonstrate this connection we will try to follow the notation used in \cite{PM_Codes}. 

In the design of the proposed coding scheme, we chose a design parameter $\mu$, and the required flexibility degree $\delta$. All the other parameters of the code including $\alpha$, $F$, $k$, $D=\{d_{1}, \cdots, d_{\delta} \}$, and $\beta(d_i)$ will be then determined based on $\mu$, and $\delta$ as follows. The per-node storage capacity is
\begin{align}\label{eq_alpha_design}
\alpha = \mu\cdot \mathrm{lcm}\left(1, \cdots , \delta \right).
\end{align}
Moreover, we have $k=\mu+1$, and $F = (\mu+1)\alpha$, which satisfies the MDS property. Finally, for $D$ we have
\begin{align}\label{eq_D_design}
D=\{d_{1}, \cdots, d_{\delta}\},~~d_{i} = (i+1)\mu, ~i\in \{1, \cdots, \delta \}.
\end{align}
and for any $d_{i} \in D$, the associated per-node and total repair bandwidths denoted by $\beta(d_{i})$, and $\gamma(d_{i})$ respectively are
\begin{align}\label{eq_beta_design}
\beta(d_{i}) = \frac{\alpha}{i \mu},~~~~\gamma(d_{i})=d_{i}\beta(d_{i})=\frac{(i+1)\alpha}{i}.
\end{align}

\subsection{Coding for Storage}\label{SubSec_CodingScheme}
We begin the introduction of the coding scheme by describing the process of encoding the source symbols and deriving the encoded symbols to be stored in the storage nodes. Similar to the product matrix codes, the first step in encoding for storage in this scheme is to arrange the information symbols in a matrix, denoted by $M$, which we refer to hereafter as the \emph{data matrix}. Let
\begin{align}
z_{\delta}=\mathrm{lcm}\left(1, \cdots , \delta \right). \nonumber
\end{align}
The data matrix in our coding scheme is structured as follows,
\begin{align}\label{eq_M_define}
M = \left[\begin{array}{c c c c c c c c}
S_{1} ~ & ~ S_{2} ~ & ~ O ~ & ~ O ~ & ~ O ~ & ~ O ~ & \cdots & O \\
S_{2} ~ & ~ S_{3} ~ & ~ S_{4} ~ & ~ O ~ & ~ O ~ & ~ O ~ & \cdots & O \\
O ~ & ~ S_{4} ~ & ~ S_{5} ~ & ~ S_{6} ~ & ~ O ~ & ~ O ~ & \cdots & O \\
O ~ & ~ O ~ & ~ S_{6} ~ & ~ S_{7} ~ & ~ S_{8} ~ & ~ O ~ & \cdots & O \\
\vdots ~ & & & & & \ddots & & \vdots \\
O ~ & ~ \cdots & & & O & S_{2z_{\delta}-4} & S_{2z_{\delta}-3} & S_{2z_{\delta}-2} \\
O ~ & ~ \cdots & & & O & O & S_{2z_{\delta}-2} & S_{2z_{\delta}-1} \\
O ~ & ~ \cdots & & & O & O & O & S_{2z_{\delta}}
\end{array}\right],
\end{align}
where, each $S_{i},~i \in \{1, \cdots, 2z_{\delta}\}$ is a symmetric $\mu\times \mu$ matrix filled with $\mu(\mu+1)/2$ source symbols, and $O$ is a $\mu \times \mu$ zero matrix. Therefore, $M$'s dimensions are $(z_{\delta}+1)\mu \times z_{\delta}\mu$. Note that the total number of distinct source symbols is
\begin{align}
F = 2 z_{\delta}\frac{\mu(\mu+1)}{2} = k\alpha. \nonumber
\end{align}



The source encoder then creates the vector of coded symbols for each of the $n$ storage nodes, by calculating the product of a node-specific coefficient vector and the data matrix. To describe this process, we first need the following definition.
\begin{defn}[Generalized Vandermonde Matrix]
For distinct, non-zero elements $e_1,\cdots,e_m$ in $\mathbb{F}_q$, and some integer $c \geq 0$, a matrix $A_{m \times \ell}$ with entries $A_{i,j} = e_i^{c+j-1}$ is referred to as a \emph{generalized Vandermonde} matrix.
\end{defn}

In particular, for distinct, non-zero elements $e_{i}$'s in $\mathbb{F}_{q}$, with $i\in \{1, \cdots, n\}$ we define a generalized Vandermonde matrix of size $n\times (z_{\delta}+1)\mu$ as
\begin{align}
\Psi = \left[\begin{array}{c c c c c}
e_{1}~ & ~e_{1}^{2}~ & \cdots & ~e_{1}^{(z_{\delta}+1)\mu} \\
e_{2}~ & ~e_{2}^{2}~ & \cdots & ~e_{2}^{(z_{\delta}+1)\mu} \\
~ & ~~ & \vdots &  \\
e_{n}~ & ~e_{n}^{2}~ & \cdots & ~e_{n}^{(z_{\delta}+1)\mu}
\end{array}\right]. \nonumber
\end{align}
Note that submatrices of $\Psi$ are also generalized Vandermonde matrices. Moreover, one can show that any square generalized Vandermonde matrix is invertible \cite{BAMSR}. 

We denote the $j^{\text{th}}$ row of $\Psi$ by $\underline{\psi}_{j}$. Then the vector of encoded symbols to be stored on node $j,~j \in \{1, \cdots, n\}$, denoted by $\underline{x}_{j}$, is calculated as
\begin{align}
\underline{x}_{j}=\underline{\psi}_{j}M. \nonumber
\end{align}

Note that the per-node storage capacity requirement for this coding scheme is then $z_{\delta}\mu$ as given by (\ref{eq_alpha_design}).

\subsection{Data Reconstruction}

In order to reconstruct all the information stored in the system, the data collector accesses $k$ arbitrary nodes in the network and downloads all their contents. To describe the details of the decoding we use the following lemma.

\begin{lem}\label{Lem_PMreconstruction}
Let $X$ and $\Psi$ be two known generalized Vandermonde matrices of size $(\mu+1) \times \mu$, and $\Delta$ be a known diagonal matrices of size $(\mu+1) \times (\mu+1)$ with non-zero distinct diagonal elements. Then one can uniquely solve the equation 
\begin{align}
X = \Psi A + \Delta \Psi B, \nonumber
\end{align}
for unknown $\mu \times \mu$ symmetric matrices $A$, and $B$.
\end{lem}
 
The proof of this lemma is presented in \cite{BAMSR}. The following theorem explains the data reconstruction procedure.

\begin{thm}\label{Thm_DataReconstruction}
For the coding scheme presented in subsection \ref{SubSec_CodingScheme}, there exists a decoding scheme to reconstruct all source symbols arranged in the data matrix $M$ from the encoded content of any arbitrary group of $k=\mu+1$ storage nodes.
\end{thm}

\begin{proof}
Let's assume the set of accessed nodes is $\{\ell_{1}, \cdots, \ell_{k}\}$. Moreover, let's denote the $k\times (z_{\delta}+1)\mu$ submatrix of $\Psi$ associated with the nodes $\ell_{1}, \cdots, \ell_{k}$, by $\Psi_{\text{DC}}$. We will further denote the submatrix of $\Psi_{\text{DC}}$ consisting of columns $(i-1)\mu+1$ through $i \mu$, by $\Psi_{\text{DC}}(i)$. In other words, we have a partitioning of $\Psi_{\text{DC}}$'s columns as
\begin{align}
\Psi_{\text{DC}} = \left[\Psi_{\text{DC}}(1), \cdots, \Psi_{\text{DC}}(z_{\delta}+1) \right]. \nonumber
\end{align}
As a result, defining the diagonal matrix
\begin{align}
\Lambda_{\text{DC}} = \left[\begin{array}{c c c c c}
e_{\ell_{1}}^{\mu}~ & ~0~ & ~0~ & \cdots & ~0 \\
0~ & ~e_{\ell_{2}}^{\mu}~ & ~0~ & \cdots & ~0 \\
\vdots ~& ~ & ~ & \ddots & ~\vdots \\
0~ & ~0~ & ~0~ & \cdots & ~e_{\ell_{k}}^{\mu} 
\end{array} \right], \nonumber
\end{align}
we have
\begin{align}\label{eq_PsiDC}
\Psi_{\text{DC}}(i+1) = \Lambda_{\text{DC}}\Psi_{\text{DC}}(i).
\end{align}
Similarly, let's denote the matrix consisting of the collected encoded vectors by $X_{\text{DC}}$, and its partitioning to $k\times \mu$ submatrices $X_{\text{DC}}(i)$, $i\in\{1,\cdots,z_{\delta}\}$ as follows
\begin{align}
X_{\text{DC}} = \left[\begin{array}{c}
\underline{x}_{\ell_{1}} \\
\vdots \\
\underline{x}_{\ell_{k}}
\end{array} \right] = \left[ X_{\text{DC}}(1), \cdots, X_{\text{DC}}(z_{\delta}) \right]. \nonumber
\end{align}

The decoding procedure for data reconstruction consists of $z_{\delta}$ consecutive steps. The first step uses only the submatrix $X_{\text{DC}}(1)$. Similar to the data reconstruction for product matrix MSR codes, using (\ref{eq_PsiDC}) we have,
\begin{align}
X_{\text{DC}}(1) &= \left[\Psi_{\text{DC}}(1), \Psi_{\text{DC}}(2) \right]\left[\begin{array}{c}
S_{1} \\
S_{2}
\end{array}\right] \nonumber \\ 
&= \Psi_{\text{DC}}(1)S_{1} + \Lambda_{\text{DC}}\Psi_{\text{DC}}(1)S_{2}. \nonumber
\end{align}

Using Lemma \ref{Lem_PMreconstruction}, the decoder recovers both $S_{1}$, and $S_{2}$, using $X_{\text{DC}}(1)$, in step one. Then, for $i\in\{2, \cdots, z_{\delta}\}$, the decoder performs step $i$ by using submatrix $X_{\text{DC}}(i)$, and decodes submatrices $S_{2i-1}$, and $S_{2i}$, as follows. 

In step $i$ of the data reconstruction decoding, having the submatrix $S_{2(i-1)}$ already recovered from step $i-1$, the decoder first calculates
\begin{align}\label{eq_X_i_calc}
\hat{X}_{\text{DC}}(i) &= X_{\text{DC}}(i) - \Psi_{\text{DC}}(i-1) S_{2(i-1)}. \nonumber \\
&= \left[\Psi_{\text{DC}}(i), \Psi_{\text{DC}}(i+1) \right]\left[\begin{array}{c}
S_{2i-1} \\
S_{2i}
\end{array}\right].
\end{align}
Then from (\ref{eq_PsiDC}), and (\ref{eq_X_i_calc}), we have
\begin{align}
\hat{X}_{\text{DC}}(i) = \Psi_{\text{DC}}(i) S_{2i-1} + \Lambda_{\text{DC}}\Psi_{\text{DC}}(i) S_{2i}. \nonumber
\end{align}
Again using Lemma \ref{Lem_PMreconstruction}, decoder recovers $S_{2i-1}$, and $S_{2i}$ at the end of the step $i$ of the decoding. Hence, by finishing step $z_{\delta}$, decoder reconstructs all the submatrices in $M$. 
\end{proof}

\subsection{Bandwidth Adaptive Exact Repair}

We now describe the bandwidth adaptive repair procedure, by assuming that node $f$ is failed and the set of helpers selected for the repair are $\mathcal{H} = \{\ell_{1}, \cdots, \ell_{d}\}$, for some arbitrary $d\in D$. The following theorem describes the repair procedure in this bandwidth adaptive MSR code.

\begin{thm}\label{Thm_BWARepair}
Consider the coding scheme presented in subsection \ref{SubSec_CodingScheme}, with design parameters $\mu$, and $\delta$, and $D$ as defined in (\ref{eq_D_design}). For any arbitrary failed node $f$, and any arbitrary set of helpers $\mathcal{H}=\{\ell_{1}, \cdots, \ell_{d}\}$, for some $d\in D$, there exists a repair scheme for recovering the content of node $f$ with per-node repair bandwidth, 
\begin{align}\label{eq_BWAbeta}
\beta(d) = \frac{\alpha}{d-\mu}.
\end{align}
\end{thm}

%
\begin{proof}
Without loss of generality let $d=(m+1)\mu$, for some $m\in\{1, \cdots, \delta\}$. Each helper node $h\in \mathcal{H}$, creates $\beta(d)=\alpha/(d-\mu)$ repair symbols to repair node $f$ as follows. First $h$ partitions its encoded content into $\beta(d)$ equal segments as 
\begin{align}\label{eq_xh_repair_partition}
\underline{x}_{h} = \left[\underline{x}_{h}(1), \cdots, \underline{x}_{h}(\beta(d)) \right].
\end{align}
Note that (\ref{eq_alpha_design}), and (\ref{eq_D_design}) guarantee that for any $d\in D$, $\alpha$ is an integer multiple of $d-\mu$, hence $\beta(d)$ is an integer. Each segment $\underline{x}_{h}(i)$ is then of size $d-\mu = m \mu$. Similarly, we split the first $\alpha$ entries of a coefficient vector assigned to node $\ell$, namely $\underline{\psi}_{\ell}$, into $\beta(d)$ equal segments as
\begin{align}\label{eq_psi_repair_partition}
\underline{\psi}_{\ell}(1:\alpha) = \left[ \underline{\psi}_{\ell}(1), \cdots, \underline{\psi}_{\ell}(\beta(d)) \right],
\end{align}
where each segment $\underline{\psi}_{\ell}(i)$ is of size $d-\mu = m \mu$.

Now each helper node $h\in \mathcal{H}$, creates its repair symbols as
\begin{align}\label{eq_rh_define}
\underline{r}(h,f) &= \left[ r_{1}(h,f), \cdots, r_{\beta(d)}(h,f)\right] \nonumber \\
&= \left[\underline{x}_{h}(1) \left(\underline{\psi}_{f}(1) \right)^{\intercal}, \cdots, \underline{x}_{h}(\beta(d)) \left(\underline{\psi}_{f}(\beta(d)) \right)^{\intercal} \right].
\end{align}

The repair decoder then receives a $d\times \beta(d)$ matrix
\begin{align}
\Upsilon_{\mathcal{H}} = \left[\begin{array}{c}
\underline{r}(\ell_{1},f) \\
\vdots \\
\underline{r}(\ell_{d},f)
\end{array} \right]. \nonumber
\end{align}
We then introduce the following partitioning of the matrix $\Upsilon_{\mathcal{H}}$, into $\beta(d)$ submatrices, as follows
\begin{align}\label{eq_Upsilon_partition}
\Upsilon_{\mathcal{H}} = \left[ \Upsilon_{\mathcal{H}}(1), \cdots, \Upsilon_{\mathcal{H}}(\beta(d)) \right],
\end{align}
where $\Upsilon_{\mathcal{H}}(i)$, $i\in\{1,\cdots,\beta\}$ is the $i^{\text{th}}$ column of $\Upsilon_{\mathcal{H}}$.

Before starting to describe the repair decoding procedure, we need to introduce some notations associated to a given repair scenario. Consider a repair procedure with $d=(m+1)\mu$, $d\in D$. For the corresponding $\beta(d)=\alpha/(d-\mu)$ we will partition matrix $M$ as depicted in Fig. \ref{Fig_M_Repair_Partition}. Note that this results in $\beta(d)$ non-overlapping diagonal submatices $M_{i}$, $i \in \{1, \cdots, \beta(d)\}$, each of size $m\mu \times m\mu$, along with $\mu \times \mu$ submatrices $S_{2m}, S_{4m}, \cdots, S_{2\beta(d)m}=S_{2z_{\delta}}$ as shown in the figure. From the construction of the data matrix, introduced in (\ref{eq_M_define}), each $M_{i}$ submatrix will be symmetric. As a result, the data matrix $M$ could be interpreted in terms of the submatrices $M_{i}$, and $S_{2i}$ for $i\in \{1, \cdots, \beta(d)\}$, associated to a repair procedure with $d=(m+1)\mu$, $d\in D$.
\begin{figure*}[!h]
\centering
\resizebox{7.2 in}{!}{
\includegraphics[scale=1]{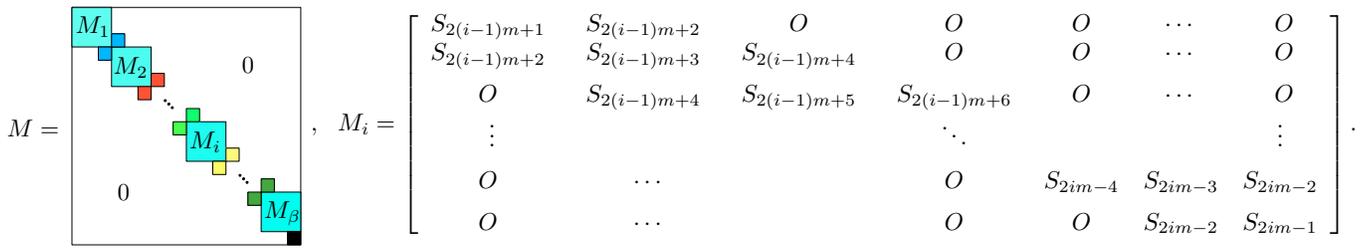}}
\caption{In the above figure $\beta$ represents $\beta(d)$. Moreover, each of the small coloured squares represent a non-zero submatrix of $M$; blue: $S_{2m}$, red: $S_{4m}$, light green: $S_{2(i-1)m}$, yellow: $S_{2im}$, dark green: $S_{2(\beta(d)-1)m}$, and black: $S_{2z_{\delta}}$.}
\label{Fig_M_Repair_Partition}
\hrulefill
\end{figure*}

In addition, for any node $\ell$, we introduce the vector $\underline{\phi}_{\ell}$ as,
\begin{align}\label{eq_phi_ell}
\underline{\phi}_{\ell} = \left[e_{\ell}~, ~\cdots~, ~e_{\ell}^{\mu} \right].
\end{align}

Finally the last notations we use to describe the adaptive repair decoding scheme, using a given set of helpers $\mathcal{H}=\{\ell_{1}, \cdots, \ell_{d}\}$, is,
\begin{align}\label{eq_Omega_i}
\Omega_{\mathcal{H}}(i) = \left[ \begin{array}{c}
\left[ \underline{\psi}_{\ell_{1}}(i), ~e_{\ell_{1}}^{i m \mu}\underline{\phi}_{\ell_{1}} \right] \\
\vdots \\
\left[ \underline{\psi}_{\ell_{d}}(i), ~e_{\ell_{d}}^{i m \mu}\underline{\phi}_{\ell_{d}} \right]
\end{array} \right], ~~i\in \{1, \cdots, \beta\}.
\end{align}
Note that, $\Omega_{\mathcal{H}}(i),~i \in\{1, \cdots, \beta(d)\}$, is a $d \times d$ generalized Vandermonde matrix and hence is invertible. We also use the following notations for submatrices of the inverse of $\Omega_{\mathcal{H}}(i)$,
\begin{align}\label{eq_Omega_partition}
\left(\Omega_{\mathcal{H}}(i)\right)^{-1} = \left[ \begin{array}{c}
\Theta_{\mathcal{H}}(i) \\
\Xi_{\mathcal{H}}(i)
\end{array}\right],
\end{align}
where $\Theta_{\mathcal{H}}(i)$ represents the top $(d-\mu) \times d$ submatrix, and $\Xi_{\mathcal{H}}(i)$, the bottom $\mu \times d$ submatrix.

The decoding procedure for the repair of node $f$ is performed in $\beta(d)$ sequential steps. In the first step, the decoder only uses the first repair symbol received from each of the helpers, namely $r_{1}(\ell_{i},f)$, for $i\in \{1, \cdots, d\}$.

Using (\ref{eq_xh_repair_partition}), (\ref{eq_psi_repair_partition}), (\ref{eq_Upsilon_partition}), and (\ref{eq_Omega_i}), and the partitioning denoted in Fig. \ref{Fig_M_Repair_Partition}, the submatrix $\Upsilon_{\mathcal{H}}(1)$, introduced in (\ref{eq_Upsilon_partition}) can be written as
\begin{align}\label{eq_repair_step1}
\Upsilon_{\mathcal{H}}(1) = \Omega_{\mathcal{H}}(1) \left[\begin{array}{c c c c}
\block(4,4){\left[\begin{array}{c c c c}
\vspace{7pt} ~~~~ & ~~ & ~~ & ~~~~ \\
 ~~~ & ~~~ & M_{1}~ & ~~~~~~ \\
\vspace{7pt} ~~~~ & ~~ & ~~ & ~~~~ 
\end{array}\right]} \\
&&& \\
&&& \\
&&& \\
~O~ & ~\cdots ~ & ~O & S_{2m}
\end{array} \right] \left( \underline{\psi}_{f}(1)\right)^{\intercal}. 
\end{align}
Multiplying the inverse of $\Omega_{\mathcal{H}}(1)$ from right to the both sides of (\ref{eq_repair_step1}), and using (\ref{eq_Omega_partition}) the decoder derives
\begin{align}\label{eq_M1_repair}
M_{1} \left( \underline{\psi}_{f}(1) \right)^{\intercal} = \Theta_{\mathcal{H}}(1) \Upsilon_{\mathcal{H}}(1),
\end{align}
 and similarly, using (\ref{eq_phi_ell}),
\begin{align}\label{eq_S2m_repair}
S_{2m} \left( e_{f}^{(m-1)\mu} \underline{\phi}_{f} \right)^{\intercal} = \Xi_{\mathcal{H}}(1) \Upsilon_{\mathcal{H}}(1).
\end{align} 
Since both $M_{1}$, and $S_{2m}$ are symmetric, from (\ref{eq_M1_repair}) we have,
\begin{align}\label{eq_additiveterm1}
\underline{\psi}_{f}(1) M_{1} = \left(\Theta_{\mathcal{H}}(1) \Upsilon_{\mathcal{H}}(1)\right)^{\intercal},
\end{align}
and from (\ref{eq_S2m_repair}), by multiplying the scalar $e_{f}^{\mu}$, we get
\begin{align}\label{eq_additiveterm2}
e_{f}^{m\mu}\underline{\phi}_{f} S_{2m}  = e_{f}^{\mu} \left( \Xi_{\mathcal{H}}(1) \Upsilon_{\mathcal{H}}(1) \right)^{\intercal}.
\end{align} 
Using a partitioning similar to (\ref{eq_xh_repair_partition}) for $\underline{x}_{f}$, from (\ref{eq_additiveterm1}), and (\ref{eq_additiveterm2}) the decoder then recovers $\underline{x}_{f}(1)$ as,
\begin{align}
\underline{x}_{f}(1) = \underline{\psi}_{f}(1) M_{1} + \left[O, \cdots, O, e_{f}^{m\mu}\underline{\phi}_{f}S_{2m} \right]_{\mu \times m\mu}, \nonumber
\end{align}
where, the rightmost term in the above expression is derived by padding $m-1$, $\mu \times \mu$ zero matrices, $O$, to the left of the matrix calculated in (\ref{eq_additiveterm2}).

In step $i$ for $i=2$ through $\beta(d)$ of the repair decoding, the decoder then recovers $\underline{x}_{f}(i)$, using $\Upsilon_{\mathcal{H}}(i)$ received from the helpers, along with $e_{f}^{(i-1)m\mu} \underline{\phi}_{f} S_{2(i-1)m}$, recovered from the step $i-1$ of decoding. To this end, the decoder first removes the contribution of the $S_{2(i-1)m}$ submatrix in the repair symbols in $\Upsilon_{\mathcal{H}}(i)$ by calculating 
\begin{align}
\hat{\Upsilon}_{\mathcal{H}}(i) \hspace{-1mm} = \hspace{-1mm} \Upsilon_{\mathcal{H}}(i) \hspace{-1mm} - \hspace{-1mm} \left[\begin{array}{c}
e_{\ell_{1}}^{(i-1)m\mu -\mu} \underline{\phi}_{\ell_{1}} \\
\vdots \\
e_{\ell_{d}}^{(i-1)m\mu -\mu} \underline{\phi}_{\ell_{d}}
\end{array} \right] \hspace{-1mm} S_{2(i-1)m} \hspace{-1mm} \left(e_{f}^{(i-1)m\mu}\underline{\phi}_{f}\right)^{\intercal}. \nonumber
\end{align}
In the above expression, $S_{2(i-1)m} \left(e_{f}^{(i-1)m\mu}\underline{\phi}_{f}\right)^{\intercal}$ is itself derived by transposing $e_{f}^{(i-1)m\mu} \underline{\phi}_{f} S_{2(i-1)m}$. Hence we have,
\begin{align}
\hat{\Upsilon}_{\mathcal{H}}(i) = \Omega_{\mathcal{H}}(i) \left[\begin{array}{c c c c}
\block(4,4){\left[\begin{array}{c c c c}
\vspace{7pt} ~~~~ & ~~ & ~~ & ~~~~ \\
 ~~~ & ~~~ & M_{i}~ & ~~~~~~ \\
\vspace{7pt} ~~~~ & ~~ & ~~ & ~~~~ 
\end{array}\right]} \\
&&& \\
&&& \\
&&& \\
~O~ & ~\cdots ~ & ~O & S_{2im}
\end{array} \right] \left( \underline{\psi}_{f}(i)\right)^{\intercal}. \nonumber
\end{align}
Therefore, similar to (\ref{eq_M1_repair}) through (\ref{eq_additiveterm2}) the decoder derives,
\begin{align}\label{eq_Mi_repair}
\underline{\psi}_{f}(i)M_{i} = \left(\Theta_{\mathcal{H}}(i) \Upsilon_{\mathcal{H}}(i)\right)^{\intercal},
\end{align}
and
\begin{align}\label{eq_additiveterm2i}
e_{f}^{im\mu} \underline{\phi}_{f} S_{2im} = e_{f}^{\mu} \left( \Xi_{\mathcal{H}}(i) \Upsilon_{\mathcal{H}}(i)\right)^{\intercal}.
\end{align} 
Finally, using (\ref{eq_Mi_repair}) and (\ref{eq_additiveterm2i}), we have
\begin{align}
\underline{x}_{f}(i) = \underline{\psi}_{f}(i) M_{i} + \left[O, \cdots, O, e_{f}^{ima}\underline{\phi}_{f}S_{2im} \right]_{\mu \times m\mu}. \nonumber
\end{align} 
\end{proof}

\begin{rmrk}
In a DSS with $n$ nodes, for $D=\{d_{1},\cdots, d_{\delta}\}$, the bandwidth adaptive MSR codes presented in \cite{ExplicitMSR}, although support any rate, require 
\begin{align}\label{eq_Ye_Barg_alpha}
\alpha = \left(\mathrm{lcm}\left(d_{1}-k+1, \cdots, d_{\delta}-k+1 \right)\right)^n. 
\end{align}
Comparing (\ref{eq_Ye_Barg_alpha}) with (\ref{eq_alpha_design}), one could see that the presented scheme reduces the required $\alpha$ (and $\beta$) values to the $n^{\text{th}}$ root. However, this scheme works only for $2k-1<d_{i},~\forall{d_{i}\in D}$. Hence, the design of high-rate bandwidth adaptive MSR codes with small $\alpha$ and $\beta$ still remains an open problem.
\end{rmrk}


\section{Conclusion}

We presented an alternative solution for exact-repair MSR codes in which optimal exact repair is guaranteed simultaneously with a range of choices, $D=\{d_1,\cdots,d_\delta\}$, for the number of helpers. Comparing to the only other explicit constructions, presented in \cite{ExplicitMSR}, we showed that when $d_{i}\geq 2k-1,~\forall{d_{i} \in D}$, the required values for $\alpha$, and $\beta$ are reduced to the $n^{\text{th}}$ root for a DSS with $n$ nodes. 
\bibliographystyle{IEEEtran}
\bibliography{IEEEabrv,DSS_Bibliography}

\end{document}